%% file: rappor.tex
\documentclass[letterpaper,twocolumn,10pt]{article}
\usepackage{usenix,epsfig,endnotes}
\AtBeginDocument{}
\usepackage{calrsfs}
\DeclareMathAlphabet{\pazocal}{OMS}{zplm}{m}{n}
\usepackage{epsfig,endnotes,xspace,hyperref,url}
\usepackage{amsmath}
\usepackage{amssymb,bbm}
\usepackage{enumerate}
\usepackage{varioref}
\usepackage{graphicx}
\usepackage{epstopdf}
\usepackage{xcolor}
\usepackage{multirow}
\usepackage{subfigure}
\usepackage{comment}
\usepackage{multirow}
\usepackage{color}
\usepackage{wrapfig}
\usepackage{enumitem}

\newtheorem{theorem}{Theorem}
\newtheorem{proof}{Proof}

\newtheorem{definition}{Definition}

\newtheorem{protocol}{Protocol}

\newcommand{\says}[2]{{\color{blue}{#1 says: }{#2}}\xspace}

\newcommand{\myexp}[1]{\ensuremath{e^{#1}}\xspace}

\renewcommand{\Pr}[1]{\ensuremath{\mathsf{Pr}\left[#1\right]}\xspace}

\newcommand{\EV}[1]{\ensuremath{\mathrm{E}\left[\, #1 \,\right]}}
\newcommand{\Var}[1]{\ensuremath{\mathrm{Var}\left[\, #1 \,\right]}}

\newcommand{\Lap}[1]{\ensuremath{\mathsf{Lap}\left(#1\right)}\xspace}

\newcommand{\mypara}[1]{\vspace*{0.05in}\noindent \textbf{#1}}

\renewcommand{\AA}{\mathbf{A}}
\newcommand{\rappor}{\noindent RAPPOR\xspace}
\newcommand{\answerSpace}{[d]}
\newcommand{\report}{\textbf{y}}
\newcommand{\reportM}{\textbf{m}}

\newcommand{\protRandom}{Direct Encoding}
\newcommand{\protHistogram}{Histogram Encoding}
\newcommand{\DE}{\ensuremath{\mathrm{DE}}\xspace}
\newcommand{\SHE}{\ensuremath{\mathrm{SHE}}\xspace}
\newcommand{\THE}{\ensuremath{\mathrm{THE}}\xspace}
\newcommand{\HE}{\ensuremath{\mathrm{HE}}\xspace}
\newcommand{\UE}{\ensuremath{\mathrm{UE}}\xspace}
\newcommand{\SUE}{\ensuremath{\mathrm{SUE}}\xspace}
\newcommand{\OUE}{\ensuremath{\mathrm{OUE}}\xspace}
\newcommand{\BLH}{\ensuremath{\mathrm{BLH}}\xspace}
\newcommand{\LH}{\ensuremath{\mathrm{LH}}\xspace}
\newcommand{\OLH}{\ensuremath{\mathrm{OLH}}\xspace}

\newcommand{\ELP}{\ensuremath{\mathrm{ELP}}\xspace}

\newcommand{\PE}{\ensuremath{\mathsf{PE}}\xspace}

\newcommand{\protPartition}{Partition\xspace}
\newcommand{\protUnary}{Unary Encoding\xspace}
\newcommand{\protProjection}{Localized Binary Hashing\xspace}

\newcommand{\ldpsys}{\textbf{\texttt{PPO}}\xspace}

\newcommand{\argmax}{\mathop{\mathrm{argmax}}}
\newcommand{\ind}{\mathbbm{1}}

\allowdisplaybreaks
\setlist[itemize]{itemsep=2pt, topsep=2pt}

\begin{document}
\sloppypar
\newcommand{\fullTitle}{Optimizing Locally Differentially Private Protocols\\
\large Accepted by Usenix'17}

\title{\fullTitle}

\author{
	{\rm Tianhao Wang, Jeremiah Blocki, Ninghui Li}\\
	Purdue University
	\and
	{\rm Somesh Jha}\\
	University of Wisconsin-Madison
} 


\maketitle

\begin{abstract}
    \input{abstract.tex}
\end{abstract}

\section{Introduction}
\input{intro.tex}

\section{Background and Existing Protocols}\label{sec:background}

\input{prelim.tex}

\input{over.tex}


\section{A Framework for LDP Protocols}\label{sec:framework}
\input{absmethod.tex}

\section{Optimizing LDP Protocols}
\label{sec:protocols}

We now cast many protocols that have been proposed into our framework of ``pure'' LDP protocols.  Casting these protocols into the framework of pure protocols enables us to derive their variance and understand how each method's accuracy is affected by parameters such as domain size, $\epsilon$, etc.  This also enables us to generalize and optimize these protocols and propose two new protocols that improve upon existing ones.  More specifically, we will consider the following protocols, which we organize by their encoding methods.

\begin{itemize}
	\item \textbf{Direct Encoding (\DE).} There is no encoding. It is a generalization of the Random Response technique.
    \item \textbf{Histogram Encoding (HE).}  An input $v$ is encoded as a histogram for the $d$ possible values.  The perturbation step adds noise from the Laplace distribution to each number in the histogram.  We consider two aggregation techniques, \SHE and \THE.
     \begin{itemize}
    	\item \textbf{Summation with Histogram Encoding (\SHE)} simply sums up the reported noisy histograms from all users.
	    \item \textbf{Thresholding with Histogram Encoding (\THE)} is parameterized by a value $\theta$; it interprets each noisy count above a threshold $\theta$ as a $1$, and each count below $\theta$ as a $0$.
    \end{itemize}
    \item \textbf{Unary Encoding (UE).}  An input $v$ is encoded as a length-$d$ bit vector, with only the bit corresponding to $v$ set to $1$.  Here two key parameters in perturbation are $p$, the probability that $1$ remains $1$ after perturbation, and $q$, the probability that $0$ is perturbed into $1$.  Depending on their choices, we have two protocols, \SUE and \OUE.
    \begin{itemize}
	  \item \textbf{Symmetric Unary Encoding (\SUE)} uses $p+q=1$; this is the Basic RAPPOR protocol~\cite{rappor}.
	  \item \textbf{Optimized Unary Encoding (\OUE)} uses optimized choices of $p$ and $q$; this is newly proposed in this paper.
    \end{itemize}
    \item \textbf{Local Hashing (LH).}  An input $v$ is encoded by choosing at random $H$ from a universal hash function family $\mathbb{H}$, and then outputting $(H,H(v))$.  This is called Local Hashing because each user chooses independently the hash function to use.  Here a key parameter is the range of these hash functions. Depending on this range, we have two protocols, \BLH and \OLH.

    \begin{itemize}
	    \item \textbf{Binary Local Hashing (\BLH)} uses hash functions that outputs a single bit.  This is equivalent to the random matrix projection technique in~\cite{Bassily2015local}.
	   \item \textbf{Optimized Local Hashing (\OLH)} uses optimized choices for the range of hash functions; this is newly proposed in this paper.
    \end{itemize}
\end{itemize}

\subsection{Direct Encoding (DE)} \label{sec:direct}
\input{random.tex}

\subsection{Histogram Encoding (HE)}
\input{laplace.tex}
\subsection{Unary Encoding (UE)}
\input{unary.tex}

\subsection{Binary Local Hashing (BLH)}\label{sec:prot:lbh}
\input{projection.tex}

\subsection{Optimal Local Hashing (OLH)} \label{sec:partition}
\input{partition.tex}


\section{Which Protocol to Use}
\label{sec:compare}
\input{variances}



\section{Experimental Evaluation}\label{sec:eval}
\input{eval.tex}

\section{Related Work}\label{sec:related}
\input{related}

\section{Conclusion}\label{sec:conclude}
\input{conc}
{
\bibliographystyle{acm}
\bibliographystyle{abbrv}
\bibliography{privacy,Ninghui,password}
}
\appendix
\section{Additional Evaluation}\label{sec:addi_eval}
\input{addi_eval.tex}

\end{document}

%% file: abstract.tex
Protocols satisfying Local Differential Privacy (LDP) enable parties
to collect aggregate information about a population while protecting
each user's privacy, without relying on a trusted third party. LDP
protocols (such as Google's RAPPOR) have been deployed in real-world
scenarios. In these protocols, a user encodes his private information
and perturbs the encoded value locally before sending it to an
aggregator, who combines values that users contribute to infer
statistics about the population.  In this paper, we introduce a
framework that generalizes several LDP protocols proposed in the
literature. Our framework yields a simple and fast aggregation
algorithm, whose accuracy can be precisely analyzed. Our in-depth
analysis enables us to choose optimal parameters, resulting in two new
protocols (i.e., Optimized Unary Encoding and Optimized Local Hashing)
that provide better utility than protocols previously proposed.  We
present precise conditions for when each proposed
protocol should be used, and perform experiments that demonstrate the advantage of our proposed protocols. 

%% file: intro.tex

Differential privacy~\cite{Dwo06,DMNS06} has been increasingly
accepted as the \textit{de facto} standard for data privacy in the
research community.  While many differentially
private algorithms have been developed for
data publishing and analysis~\cite{dpbook,Li2016book}, there have been few deployments of
such techniques.  Recently, techniques for satisfying differential
privacy (DP) in the local setting, which we call {LDP}, have been
deployed.  Such techniques enable gathering of statistics while
preserving privacy of every user, without relying on trust in a single
data curator.  For example, researchers from Google
developed \rappor~\cite{rappor,rappor2}, which is included as part of
Chrome.  It enables Google to collect users' answers to questions such as
the default homepage of the browser, the default search engine, and so on, 
to understand the unwanted or malicious hijacking of user settings.
Apple~\cite{Apple} also uses similar methods to
help with predictions of spelling and other things, but the details of
the algorithm are not public yet.  Samsung proposed a similar
system~\cite{samsung} which enables collection of not only categorical
answers (e.g., screen resolution) but also numerical answers (e.g.,
time of usage, battery volume), although it is not clear whether this
has been deployed by Samsung.

A basic goal in the LDP setting is frequency estimation.  A protocol
for doing this can be broken down into following steps: For each
question, each user \textbf{encodes} his or her answer (called input) into a specific format, \textbf{randomizes} the
encoded value to get an output, and then sends the output to the aggregator, who then \textbf{aggregates} and decodes
the reported values to obtain, for each value of interest, an
estimate of how many users have that value.

We introduce a framework for what we call ``pure'' LDP protocols, which has a nice symmetric property.  
We introduce a simple, generic aggregation and decoding technique that works for
all pure LDP protocols, and prove that this technique results in an
unbiased estimate.  We also present a formula for the variance of the
estimate.  Most existing protocols fit our proposed framework.  The framework
 also enables us to precisely analyze and compare the accuracy of
different protocols, and generalize and optimize them.  For example,
we show that the Basic RAPPOR protocol~\cite{rappor}, which essentially uses unary
encoding of input, chooses sub-optimal parameters for the
randomization step.  Optimizing the parameters results in what we call
the Optimized Unary Encoding (\OUE) protocol, which has significantly
better accuracy.

Protocols based on unary encoding require $\Theta(d)$ communication cost, where $d$ is the number of possible
input values, and can be very large (or even unbounded) for some applications. The RAPPOR protocol
uses a Bloom filter encoding to reduce the communication cost; however,
this comes with a cost of decreasing accuracy as well as increasing computation
cost for aggregation and decoding. The random matrix projection-based approach introduced
in~\cite{Bassily2015local} has $\Theta(\log n)$ communication cost
(where $n$ is the number of users); however, its accuracy is unsatisfactory.
We observe that in our framework this protocol can be interpreted as binary
local hashing.  Generalizing this and optimizing the parameters
results in a new Optimized Local Hashing (\OLH) protocol, which
provides much better accuracy while still
requiring $\Theta(\log n)$ communication cost.  The variance of \OLH
is orders of magnitude smaller than the previous methods, for
$\epsilon$ values used in RAPPOR's implementation.
Interestingly, \OLH has the same error variance as \OUE; thus it reduces communication cost at no cost of utility.

With LDP, it is possible to collect data that was inaccessible because of privacy issues.
Moreover, the increased amount of data will significantly improve the performance of some learning tasks.
Understanding customer statistics help cloud server and software platform operators to
better understand the needs of populations and offer more effective and reliable services.
Such privacy-preserving crowd-sourced statistics are also useful for providing better security while maintaining a level of privacy.
For example, in~\cite{rappor}, it is demonstrated that such techniques can be applied to collecting windows process names and
Chrome homepages to discover malware processes and unexpected default homepages (which could be malicious).

Our paper makes the following contributions:
\begin{itemize}
	\item We introduce a framework for ``pure'' LDP protocols, and
	develop a simple, generic aggregation and decoding technique
	that works for all such protocols.  This framework enables us
	to analyze, generalize, and optimize different LDP protocols.
	
	\item We introduce the Optimized Local Hashing (\OLH)
    protocol, which has low communication cost and provides much
    better accuracy than existing protocols.  For $\epsilon \approx
    4$, which was used in the RAPPOR implementation, the variance
    of \OLH's estimation is $1/2$ that of
    RAPPOR, and close to $1/14$ that of Random
    Matrix Projection~\cite{Bassily2015local}. Systems using LDP as a primitive could benefit significantly by adopting improved LDP protocols like \OLH.
	
\end{itemize}


\textbf{Roadmap.} In Section~\ref{sec:background}, we describe existing
protocols from~\cite{rappor,Bassily2015local}.  We then present our
framework for pure LDP protocols in Section~\ref{sec:framework}, apply
it to study LDP protocols in Section~\ref{sec:protocols}, and compare
different LDP protocols in Section~\ref{sec:compare}.  We show experimental results in Section~\ref{sec:eval}.
We discuss related work in
Section~\ref{sec:related} and conclude in Section~\ref{sec:conclude}.

%% file: prelim.tex
The notion of differential privacy was originally introduced for the setting where there is a \textbf{trusted data curator}, who gathers data from individual users, processes the data in a way that satisfies DP, and then publishes the results.  Intuitively, the DP notion requires that any single element in a dataset has only a limited impact on the output.

\begin{definition}[Differential Privacy] \label{def:diff}
An algorithm $\AA$ satisfies $\epsilon$-differential privacy ($\epsilon$-DP), where $\epsilon\geq 0$,
if and only if for any datasets $D$ and $D'$ that \emph{differ in one element}, we have
\begin{equation*}
\forall{t\!\in\! \mathit{Range}(\AA)}:\; \Pr{\AA(D) = t} \leq e^{\epsilon}\, \Pr{\AA(D')= t}, 
\end{equation*}
where $\mathit{Range}(\AA)$ denotes the set of all possible outputs of the algorithm $\AA$.
\end{definition}

\subsection{Local Differential Privacy Protocols}

In the local setting, there is no trusted third party.  An \textbf{aggregator} wants to gather information from \textbf{users}.  Users are willing to help the aggregator, but do not fully trust the aggregator for privacy. For the sake of privacy, each user perturbs her own data before sending it to the aggregator (via a secure channel).
For this paper, we consider that each user has a single value $v$, which can be viewed as the user's answer to a given question.  The aggregator aims to find out the frequencies of values among the population.  Such a data collection protocol consists of the following algorithms:
\begin{itemize}

 \item
$\mathsf{Encode}$ is executed by each user.  The algorithm takes an input value $v$ and outputs an encoded value $x$.

 \item
$\mathsf{Perturb}$, which takes an encoded value $x$ and outputs $y$.  Each user with value $v$ reports  $y=\mathsf{Perturb}(\mathsf{Encode}(v))$.
For compactness, we use $\PE(\cdot)$ to denote the composition of the encoding and perturbation algorithms, i.e., $\PE(\cdot)=\mathsf{Perturb}(\mathsf{Encode}(\cdot))$.
$\PE(\cdot)$ should satisfy $\epsilon$-local differential privacy, as defined below.

 \item
$\mathsf{Aggregate}$ is executed by the aggregator; it takes all the reported values, and outputs aggregated information.
\end{itemize}




\begin{definition}[Local Differential Privacy] \label{def:dlp}
    An algorithm $\AA$ satisfies $\epsilon$-local differential privacy ($\epsilon$-LDP), where $\epsilon \geq 0$,
    if and only if for any input $v_1$ and $v_2$, we have
    \begin{equation*}
    \forall{y \in \mathit{Range}(\AA)}:\; \Pr{\AA(v_1)=y} \leq e^{\epsilon}\, \Pr{\AA(v_2)=y},
    \end{equation*}
    where $\mathit{Range}(\AA)$ denotes the set of all possible outputs of the algorithm $\AA$.
\end{definition}

This notion is related to
\emph{randomized response}~\cite{Warner65}, which is a decades-old technique in social science to collect statistical information about embarrassing or illegal behavior.  To report a single bit by random response, one reports the true value with probability $p$ and the flip of the true value with probability $1-p$.  This satisfies $\left(\ln\frac{p}{1-p}\right)$-LDP.

Comparing to the setting that requires a trusted data curator, the local setting offers a stronger level of protection, because the aggregator sees only perturbed data.  Even if the aggregator is malicious and colludes with all other participants, one individual's private data is still protected according to the guarantee of LDP.



\mypara{Problem Definition and Notations.}
There are $n$ users. Each user $j$ has one value $v^j$ and report once.
We use $d$ to denote the size of the domain of the values the users have, and $[d]$ to denote the set $\{1, 2, \ldots, d\}$.
Without loss of generality, we assume the input domain is $[d]$.
The most basic goal of $\mathsf{Aggregate}$ is \textbf{frequency estimation}, i.e., estimate, for a given value $i \in [d]$, how many users have the value $i$.
Other goals have also been considered in the literature.  One goal is, when $d$ is very large, identify values in $[d]$ that are frequent, without going through every value in $[d]$~\cite{rappor2,Bassily2015local}.
In this paper, we focus on frequency estimation.  This is the most basic primitive and is a necessary building block for all other goals.  Improving this will improve effectiveness of other protocols.

%% file: over.tex
\subsection{Basic RAPPOR}\label{sec:basic_rappor}

RAPPOR~\cite{rappor} is designed to enable longitudinal collections, where the collection happens multiple times.  Indeed, Chrome's implementation of RAPPOR~\cite{rapporlink} collects answers to some questions once every $30$ minutes. Two protocols, Basic RAPPOR and RAPPOR, are proposed in~\cite{rappor}. We first describe Basic RAPPOR.

\mypara{Encoding.}
$\mathsf{Encode}(v)=B_0$, where $B_0$ is a length-$d$ binary vector such that $B_0[v]=1$ and $B_0[i]=0$ for $i \ne v$.  We call this \textbf{Unary Encoding}.

\mypara{Perturbation.}  $\mathsf{Perturb}(B_0)$ consists of two steps:

\noindent{\it Step 1: Permanent randomized response:} Generate $B_1$ such that:
\[
\Pr{ B_1[i] = 1}  = \left\{
\begin{array}{lr}
1-\frac{1}{2}f, & \mbox{if} \; B_0[i] = 1, \\
\frac{1}{2}f, & \mbox{if} \; B_0[i] = 0. \\
\end{array}
\right.
\]
RAPPOR's implementation uses $f=1/2$ and $f=1/4$.  Note that this randomization is \textbf{symmetric} in the sense that $\Pr{B_1[i]=1|B_0[i]=1}=\Pr{B_1[i]=0|B_0[i]=0}=1-\frac{1}{2}f$; that is, the probability that a bit of $1$ is preserved equals the probability that a bit of $0$ is preserved.
This step is carried out only once for each value $v$ that the user has.  

\noindent
{\it Step 2: Instantaneous randomized response:} Report $B_2$ such that:
\[
\Pr{ B_2[i] = 1}  = \left\{
\begin{array}{lr}
p, & \mbox{if} \; B_1[i] = 1, \\
q, & \mbox{if} \; B_1[i] = 0. \\
\end{array}
\right.
\]
This step is carried out each time a user reports the value.  That is, $B_1$ will be perturbed to generate different $B_2$'s for each reporting.  RAPPOR's implementation~\cite{RAPPORcode} uses $p=0.75$ and $q=0.25$, and is hence also symmetric because $p+q=1$.

We note that as both steps are symmetric, their combined effect can also be modeled by a symmetric randomization.  Moreover, we study the problem where each user only reports once. Thus without loss of generality, we ignore the instantaneous randomized response step and consider only the permanent randomized response when trying to identify effective protocols.

\mypara{Aggregation.}  
Let $B^j$ be the reported vector of the $j$-th user.  Ignoring the Instantaneous randomized response step, to estimate the number of times $i$ occurs, the aggregator computes:
\begin{align*}
\tilde{c}(i) &= \frac{\sum_{j}\ind_{\{i\mid B^j[i] = 1\}}(i) - \frac{1}{2}fn }{1-f}
\end{align*}
That is, the aggregator first counts how many time $i$ is reported by computing $\sum_{j}\ind_{\{i\mid B^j[i] = 1\}}(i)$, which counts how many reported vectors have the $i$'th bit being $1$, and then corrects for the effect of randomization. We use $\ind_X(i)$ to denote the indicator function such that:
\[
\ind_X(i)  = \left\{
\begin{array}{lr}
1, & \mbox{if} \; i\in X, \\
0, & \mbox{if} \; i\notin X. \\
\end{array}
\right.
\]

\mypara{Cost.}  The communication and computing cost is $\Theta(d)$ for each user, and $\Theta(nd)$ for the aggregator.

\mypara{Privacy.} Against an adversary who may observe multiple transmissions, this achieves $\epsilon$-LDP for $\epsilon= \ln \left(\left(\frac{1-\frac{1}{2}f}{\frac{1}{2}f}\right)^2\right)$, which is $\ln 9$ for $f=1/2$ and $\ln 49$ for $f=1/4$.

\subsection{RAPPOR}\label{sec:rappor}

Basic RAPPOR uses unary encoding, and does not scale when $d$ is large.  To address this problem, RAPPOR uses Bloom filters~\cite{Bloom70}.
While Bloom filters are typically used to encode a set for membership testing, in RAPPOR it is used to encode a single element.

\mypara{Encoding.} Encoding uses a set of $m$ hash functions $\mathbb{H} = \{H_1, H_2, \ldots, H_m\}$, each of which outputs an integer
in $[k]=\{0,1,\ldots,k-1\}$.
$\mathsf{Encode}(v)=B_0$, which is $k$-bit binary vector such that
\[
B_0[i] = \left\{
\begin{array}{lr}
1, & \mbox{if} \;\;\;\;\; \exists H \in \mathbb{H}, s.t., H(v) = i, \\
0, & \mbox{otherwise.} \\
\end{array}
\right.
\]

\mypara{Perturbation.} The perturbation process is identical to that of Basic RAPPOR.

\mypara{Aggregation.}
The use of shared hashing creates challenges due to potential collisions.  If two values happen to be hashed to the same set of indices, it becomes impossible to distinguish them.  To deal with this problem, RAPPOR introduces the concept of cohorts.  The users are divided into a number of cohorts.  Each cohort uses a different set of hash functions, so that the effect of collisions is limited to within one cohort.
However, partial collisions, i.e., two values are hashed to overlapping (though not identical) sets of indices, can still occur and interfere with estimation.  These complexities make the aggregation algorithm more complicated.  RAPPOR uses LASSO and linear regression to estimate frequencies of values.

\mypara{Cost.}  The communication and computing cost is $\Theta(k)$ for each user.  
The aggregator's computation cost is higher than Basic RAPPOR due to the usage of LASSO and regression. 

\mypara{Privacy.}  RAPPOR achieves $\epsilon$-LDP for $\epsilon=\ln \left(\left(\frac{1-\frac{1}{2}f}{\frac{1}{2}f}\right)^{2m}\right)$.  The RAPPOR implementation uses $m=2$; thus this is $\ln 81 \approx 4.39$ for $f=1/2$ and $\ln 7^4 \approx 7.78$ for $f=1/4$.

\subsection{Random Matrix Projection}\label{sec:over:matrix}

Bassily and Smith~\cite{Bassily2015local} proposed a protocol that uses random matrix projection.  This protocol has an additional Setup step.

\mypara{Setup.} The aggregator generates a public matrix $\Phi\in\{-\frac{1}{\sqrt{m}}, \frac{1}{\sqrt{m}}\}^{m\times d}$ uniformly at random.
Here $m$ is a parameter determined by the error bound,
where the ``error'' is defined as the maximal distance between the estimation and true frequency of any domain.

\mypara{Encoding.} $\mathsf{Encode}(v)=\langle r,x\rangle$, where $r$ is selected uniformly at random from $[m]$, and $x$ is the $v$'s element of the $r$'s row of $\Phi$, i.e., $x=\Phi[r, v]$.

\mypara{Perturbation.} $\mathsf{Perturb}(\langle r,x\rangle)=\langle r,\; b\cdot c\cdot m\cdot x\rangle$, where
\begin{align*}
b &  = \left\{
\begin{array}{rl}
1  & \;\;\;\mbox{ with probability } p=\frac{e^\epsilon}{e^\epsilon+1},
 \\
-1  & \;\;\; \mbox{ with probability } q=\frac{1}{e^\epsilon+1},\\
\end{array}
\right.
\\
c & =(e^\epsilon+1)/(e^\epsilon-1).
\end{align*}

\mypara{Aggregation.}
Given reports $\langle r^j,y^j \rangle$'s, the estimate for $i \in [d]$ is given by
\begin{align*}
\tilde{c}(i) &= \sum_j y^j \cdot \Phi[r^j,i].
\end{align*}
The effect is that each user with input value $i$ contributes $c$ to $\tilde{c}(i)$ with probability $p$, and $-c$ with probability $q$; thus the expected contribution is
$$(p-q)c=\left(\frac{e^\epsilon}{e^\epsilon+1} - \frac{1}{e^\epsilon+1}\right)\cdot \frac{e^\epsilon+1}{e^\epsilon-1} = 1.$$
Because of the randomness in $\Phi$, each user with value $\ne i$ contributes to $\tilde{c}(i)$ either $c$ or $-c$, each with probability $1/2$; thus the expected contribution from all such users is $0$.
Note that each row in the matrix is essentially a random hashing function mapping each value in $[d]$ to a single bit.  Each user selects such a hash function, uses it to hash her value into one bit, and then perturb this bit using random response.


\mypara{Cost.}  A straightforward implementation of the protocol is expensive.
However, the public random matrix $\Phi$ does not need to be explicitly computed.  For example, using a common pseudo-random number generator, each user can randomly choose a seed to generate a row in the matrix and send the seed in her report.
With this technique, the communication cost is $\Theta(\log m)$ for each user, and the computation cost is $O(d)$ for computing one row of the $\Phi$.
The aggregator needs $\Theta(dm)$ to generate $\Phi$, and $\Theta(md)$ to compute the estimations.


%% file: absmethod.tex

Multiple protocols have been proposed for estimating frequencies under LDP, and one can envision other protocols.  A natural research question is how do they compare with each other? Under the same level of privacy, which protocol provides better accuracy in aggregation, with lower cost?  Can we come up with even better ones?  To answer these questions, we define a class of LDP protocols that we call ``pure''.


For a protocol to be pure, we require the specification of an additional function $\mathsf{Support}$,
which maps each possible output $y$ to a set of input values that $y$ ``supports''. For example, in the basic RAPPOR protocol, an output binary vector $B$ is interpreted as supporting each input whose corresponding bit is $1$, i.e., $\mathsf{Support}(B)=\{ i \mid B[i]=1\}$.

\begin{definition}[Pure LDP Protocols]
A protocol given by $\PE$ and $\mathsf{Support}$ is pure if and only if there exist two probability values $p^* > q^*$ such that for all $v_1$,
\begin{align*}
\Pr{\PE(v_1) \in \{ y \mid v_1 \in \mathsf{Support}(y)\}}  = p^*,
\\
\forall_{v_2 \ne v_1} \Pr{\PE(v_2) \in \{ y \mid v_1 \in \mathsf{Support}(y) \}} = q^*.
\end{align*}
\end{definition}
A pure protocol is in some sense ``pure and simple''.  For each input $v_1$, the set $\{ y \mid v_1 \in \mathsf{Support}(y)\}$ identifies all outputs $y$ that ``support'' $v_1$, and can be called the support set of $v_1$.  A pure protocol requires the probability that any value $v_1$ is mapped to its own support set be the same for all values.  In order to satisfy LDP, it must be possible for a value $v_2 \ne v_1$ to be mapped to $v_1$'s support set.  It is required that this probability must be the same for all pairs of $v_1$ and $v_2$.
Intuitively, we want $p^*$ to be as large as possible, and $q^*$ to be as small as possible.  However, satisfying $\epsilon$-LDP requires that $\frac{p^*}{q^*} \leq e^{\epsilon}$.

Basic RAPPOR is pure with $p^*=1-\frac{f}{2}$ and $q^*=\frac{f}{2}$.  RAPPOR is not pure because there does not exist a suitable $q^*$ due to collisions in mapping values to bit vectors.  Assuming the use of two hash functions, if $v_1$ is mapped to $[1,1,0,0]$, $v_2$ is mapped to $[1,0,1,0]$, and $v_3$ is mapped to $[0,0,1,1]$, then because
$[1,1,0,0]$ differs from $[1,0,1,0]$ by only two bits, and from $[0,0,1,1]$ by four bits, the probability that $v_2$ is mapped to $v_1$'s support set is higher than that of $v_3$ being mapped to $v_1$'s support set.


For a pure protocol, let $y^j$ denote the submitted value by user $j$, a simple aggregation technique to estimate the number of times that $i$ occurs is as follows:
\begin{align}
\tilde{c}(i)= \frac{\sum_j \ind_{\mathsf{Support}(y^j)} (i) - n  q^*}{p^* - q^*} \label{eq:absestimate}
\end{align}
The intuition is that each output that supports $i$ gives an count of $1$ for $i$.  However, this needs to be normalized, because even if every input is $i$, we only expect to see $n \cdot p^*$ outputs that support $i$, and even if input $i$ never occurs, we expect to see $n\cdot q^*$ supports for it.  Thus the original range of $0$ to $n$ is ``compressed'' into an expected range of $nq^*$ to $np^*$.  The linear transformation in \eqref{eq:absestimate} corrects this effect.

\begin{theorem}\label{thm:pure_unbiased}
For a pure LDP protocol $\PE$ and $\mathsf{Support}$, \eqref{eq:absestimate} is unbiased, i.e., $\forall_i \EV{\tilde{c}(i)} = n f_i$, where $f_i$ is the fraction of times that the value $i$ occurs.
\end{theorem}
\begin{proof}
\begin{align*}
	\EV{ \tilde{c}(i)} = & \EV{\frac{ \left(\sum_j \ind_{\mathsf{Support}(y^j)} (i)\right) - nq^* }{p^* - q^*}}\\
					= & \frac{nf_ip^*+n(1-f_i)q^* - nq^*}{p^*-q^*}\\
					= & n \cdot \frac{f_i p^*+q^*-f_iq^* - q^*}{p^*-q^*}\\
					=  & n f_i
\end{align*}
\end{proof}
\begin{theorem}\label{thm:pure_var}
For a pure LDP protocol $\PE$ and $\mathsf{Support}$, the variance of the estimation $\tilde{c}(i)$ in \eqref{eq:absestimate} is:
\begin{equation}
\mathrm{Var}[\tilde{c}(i)] = \frac{nq^*(1-q^*)}{(p^*-q^*)^2}  + \frac{n f_i (1-p^*-q^*)}{p^*-q^*}\label{eq:absvar_exact}
\end{equation}
\end{theorem}
\begin{proof}
The random variable $\tilde{c}(i)$ is the (scaled) summation of $n$ independent random variables drawn from the Bernoulli distribution. More specifically, $nf_i$ (resp. $(1-f_i)n$) of these random variables are drawn from the Bernoulli distribution with parameter $p^*$ (resp. $q^*$).
Thus,
{\small
\begin{align}
\mathrm{Var}[\tilde{c}(i)] &=  \mathrm{Var}\left[\frac{ \left(\sum_j \ind_{\mathsf{Support}(y^j)} (i)\right) - nq^* }{p^* - q^*}\right]   \nonumber \\
&= \frac{\sum_j \mathrm{Var}[\ind_{\mathsf{Support}(y^j)} (i)]}{(p^* - q^*)^2}  \nonumber  \\
&= \frac{n f_i p^* (1-p^*)+n(1-f_i)q^* (1-q^*)}{(p^*-q^*)^2}  \nonumber \\
&= \frac{nq^* (1-q^*)}{(p^*-q^*)^2} + \frac{n f_i (1-p^*-q^*)}{p^*-q^*} 
\end{align}
}

\end{proof}

In many application domains, the vast majority of values appear very infrequently, and one wants to identify the more frequent ones.  The key to avoid having lots of false positives is to have low estimation variances for the infrequent values.  When $f_i$ is small, the variance in \eqref{eq:absvar_exact} is dominated by the first term.  We use $\mathrm{Var^*}$ to denote this approximation of the variance, that is:
\begin{equation}
\mathrm{Var^*}[\tilde{c}(i)] = \frac{nq^*(1-q^*)}{(p^*-q^*)^2}\label{eq:absvar}
\end{equation}
We also note that some protocols have the property that $p^*+q^*=1$, in which case $\mathrm{Var^*}=\mathrm{Var}$.  

As the estimation $\tilde{c}(i)$ is the sum of many independent random variables, its distribution is very close to a normal distribution.  Thus, the mean and variance of $\tilde{c}(i)$ fully characterizes the distribution of $\tilde{c}(i)$ for all practical purposes.
When comparing different methods, we observe that fixing $\epsilon$, the differences are reflected in the constants for the variance, which is where we focus our attention.

%% file: random.tex

One natural method is to extend the binary response method to the case where the number of input values is more than $2$.
This is used in~\cite{wang2016private}.  


\mypara{Encoding and Perturbing.} $\mathsf{Encode}_{\mathrm{DE}}(v)=v$, and $\mathsf{Perturb}$ is defined as follows.
\[
\Pr{\mathsf{Perturb}_{\mathrm{DE}}(x) = i}  = \left\{
\begin{array}{lr}
p=\frac{e^\epsilon}{e^\epsilon + d - 1}, & \mbox{if} \; i = x  \\
q= \frac{1-p}{d-1}=\frac{1}{e^\epsilon + d - 1}, & \mbox{if} \; i \ne x  \\
\end{array}
\right.
\]

\begin{theorem}[Privacy of DE]
	\label{thm:ran-ldp}
	The \protRandom{} (DE) Protocol satisfies $\epsilon$-LDP.
\end{theorem}
\begin{proof}
	For any inputs $v_1,v_2$ and output $y$, we have:
	\begin{align*}
	\frac{\Pr{\PE_{\mathrm{DE}}(v_1)=y}}{\Pr{\PE_{\mathrm{DE}}(v_2)=y}}
	\leq \frac{p}{q}
	=\frac{e^\epsilon/(e^\epsilon + d - 1)}{1/(e^\epsilon + d - 1)}
	=e^\epsilon
	\end{align*}
\end{proof}

\mypara{Aggregation.}
Let the $\mathsf{Support}$ function for $\mathrm{DE}$ be $\mathsf{Support}_{\mathrm{DE}}(i)=\{i\}$, i.e., each output value $i$ supports the input $i$.  
Then this protocol is pure, with $p^*=p$ and $q^*=q$.
Plugging these values into \eqref{eq:absvar}, we have
\begin{align*}
\mathrm{Var^*}[\tilde{c}_{\mathsf{DE}}(i)] = n\cdot \frac{d-2+e^\epsilon}{(e^\epsilon -1)^2}
\end{align*}
Note that the variance given above is linear in $nd$.  As $d$ increases, the accuracy of \DE suffers.  This is because, as $d$ increases, $p=\frac{e^\epsilon}{e^\epsilon + d - 1}$, the probability that a value is transmitted correctly, becomes smaller.  For example, when $e^\epsilon=49$ and $d=2^{16}$, we have $p=\frac{49}{65584}\approx 0.00075$.

%% file: laplace.tex
In Histogram Encoding (\HE), an input $x\in [d]$ is encoded using a length-$d$ histogram.  

	

\mypara{Encoding.}  $\mathsf{Encode}_{\mathrm{HE}}(v)=[0.0,0.0,\cdots,1.0,\cdots,0.0]$, where only the $v$-th component is $1.0$.  Two different input $v$ values will result in two vectors that have $\mathsf{L1}$ distance of $2.0$.

\mypara{Perturbing.}
$\mathsf{Perturb}_{\mathrm{HE}}(B)$ outputs $B'$ such that $B'[i] = B[i] + \Lap{\frac{2}{\epsilon}}$,
where $\Lap{\beta}$ is the Laplace distribution where $\Pr{\Lap{\beta}=x} = \frac{1}{2\beta} \myexp{-|x|/\beta}$.

\begin{theorem}[Privacy of HE]
	\label{thm:hist-ldp}
	The \protHistogram{} protocol satisfies $\epsilon$-LDP.
\end{theorem}
\begin{proof}
	For any inputs $v_1,v_2$, and output $B$, we have
$$
	\begin{array}{rll}
\frac{\Pr{B|v_1}}{\Pr{B|v_2}}
	&=\frac{\prod_{i \in [d]}\Pr{B[i]|v_1}}{\prod_{i \in [d]}\Pr{B[i]|v_2}}
	&= \frac{\Pr{B[v_1]|v_1}\Pr{B[v_2]|v_1}}{\Pr{B[v_1]|v_2}\Pr{B[v_2]|v_2}}
\\
	& \leq e^{\epsilon/2} \cdot e^{\epsilon/2} & =e^\epsilon
	\end{array}
$$
\end{proof}

\mypara{Aggregation: Summation with Histogram Encoding (SHE)} works as follows: For each value, sum the noisy counts for that value reported by all users.
That is, $\tilde{c}_{\mathrm{SHE}}(i) = \sum_j B^j[i]$, where $B^j$ is the noisy histogram received from user $j$.
This aggregation method does not provide a $\mathsf{Support}$ function and is not pure.  We prove its property as follows.
\begin{theorem}
	\label{thm:hist}
	In \SHE, the estimation $\tilde{c}_{\mathrm{SHE}}$ is unbiased.  Furthermore, the variance is
\begin{align*}
\Var{\tilde{c}_{\mathrm{SHE}}(i)} = n \frac{8}{\epsilon^2}
\end{align*}
\end{theorem}
\begin{proof}
Since the added noise is $0$-mean; the expected value of the sum of all noisy counts is the true count.

The $\Lap{\beta}$ distribution has variance $\frac{2}{\beta^2}$, since $\beta=\frac{\epsilon}{2}$ for each $B^j[i]$, then the variance of each such variable is $\frac{8}{\epsilon^2}$, and the sum of $n$ such independent variables have variance $n \frac{8}{\epsilon^2}$.
\end{proof}

\mypara{Aggregation: Thresholding with Histogram Encoding (THE)} interprets a vector of noisy counts discretely by defining
$$\mathsf{Support}_{\mathrm{THE}}(B) = \{ v \mid B[v] > \theta \}$$
That is, each noise count that is $>\theta$ supports the corresponding value.
This thresholding step can be performed either by the user or by the aggregator.  It does not access the original value, and thus does not affect the privacy guarantee.
Using thresholding to provide a $\mathsf{Support}$ function makes the protocol pure.  The probability $p^*$ and $q^*$ are given by
\begin{align*}
p^* = 1 - F(\theta-1);\;\;  q^* = 1 - F(\theta),
\\
\mbox{where }F(x) = \left\{
\begin{array}{lr}
\frac{1}{2}e^{\frac{\epsilon}{2}x}, & \mbox{if} \; x < 0  \\
1 - \frac{1}{2}e^{-\frac{\epsilon}{2}x}, & \mbox{if} \; x \geq 0  \\
\end{array}
\right.
\end{align*}
Here, $F(\cdot)$ is the cumulative distribution function of Laplace distribution.
If $0\leq \theta\leq 1$, then we have
\[
p^* = 1 - \frac{1}{2}e^{\frac{\epsilon}{2}(\theta-1)};\;\; q^* = \frac{1}{2}e^{-\frac{\epsilon}{2}\theta}.
\]
Plugging these values into \eqref{eq:absvar}, we have
\begin{align*}
\mathrm{Var^*}[\tilde{c}_{\mathrm{HET}}(i)] = n\cdot \frac{2e^{\epsilon \theta/2}-1}{(1+e^{\epsilon(\theta-1/2)}-2e^{\epsilon \theta/2})^2} \label{eq:ohe:var}
\end{align*}

\mypara{Comparing $\mathrm{SHE}$ and $\mathrm{THE}$.}
Fixing $\epsilon$, one can choose a $\theta$ value to minimize the variance.  Numerical analysis shows that the optimal $\theta$ is in $(\frac{1}{2}, 1)$, and depends on $\epsilon$. When $\epsilon$ is large, $\theta\rightarrow 1$. Furthermore, $\mathrm{Var}[\tilde{c}_{\mathsf{THE}}]< \mathrm{Var}[\tilde{c}_{\mathsf{SHE}}]$ is always true.  This means that by thresholding, one improves upon directly summing up noisy counts, likely because thresholding limits the impact of noises of large magnitude.  In Section~\ref{sec:compare}, we illustrate the differences between them using actual numbers.


%% file: unary.tex

Basic RAPPOR, which we described in Section~\ref{sec:basic_rappor}, takes the approach of directly perturb a bit vector.  We now explore this method further.

	

\mypara{Encoding.}   $\mathsf{Encode}(v)=[0,\cdots,0,1,0,\cdots,0]$, a length-$d$ binary vector where only the $v$-th position is $1$.

\mypara{Perturbing.}
$\mathsf{Perturb}(B)$ outputs $B'$ as follows:
 \[
\Pr{B'[i] = 1}  = \left\{
\begin{array}{lr}
p, & \mbox{if} \; B[i] = 1  \\
q, & \mbox{if} \; B[i] = 0  \\
\end{array}
\right.
\]

\begin{theorem}[Privacy of UE]
	\label{thm:una-ldp}
	The \protUnary{} protocol satisfies $\epsilon$-LDP for
\begin{align}
 \epsilon= \ln \left(\frac{p(1-q)}{(1-p)q}\right)
  \label{eq:uereq}
\end{align}

\end{theorem}
\begin{proof}
	For any inputs $v_1,v_2$, and output $B$, we have
	\begin{align}
	\frac{\Pr{B|v_1}}{\Pr{B|v_2}}
	=&\frac{\prod_{i \in [d]}\Pr{B[i]|v_1}}{\prod_{i \in [d]}\Pr{B[i]|v_2}}\label{aln:una-ind}\\
	\leq &\frac{\Pr{B[v_1]=1|v_1}\Pr{B[v_2]=0|v_1}}{\Pr{B[v_1]=1|v_2}\Pr{B[v_2]=0|v_2}}\label{aln:una-diff}\\
	=&\frac{p}{q}\cdot\frac{1-q}{1-p}=e^\epsilon\nonumber
	\end{align}
	\eqref{aln:una-ind} is because each bit is flipped independently,
	and \eqref{aln:una-diff} is because $v_1$ and $v_2$ result in bit vectors that differ only in locations $v_1$ and $v_2$, and a vector with position $v_1$ being $1$ and position $v_2$ being $0$ maximizes the ratio.
\end{proof}

\mypara{Aggregation.}
A reported bit vector is viewed as supporting an input $i$ if  $B[i]=1$, i.e.,
$\mathsf{Support}_{\mathrm{UE}}(B)=\{i \mid B[i]=1\}$.
This yields $p^*=p$ and $q^*=q$.
Interestingly, \eqref{eq:uereq} does not fully determine the values of $p$ and $q$ for a fixed $\epsilon$.  Plugging \eqref{eq:uereq} into 
\eqref{eq:absvar}, we have
\begin{align}
\mathrm{Var^*}[\tilde{c}_{\UE}(i)] &=\frac{nq(1-q)}{(p-q)^2}\nonumber
= \frac{nq(1-q)}{(\frac{e^\epsilon q}{1-q+e^\epsilon q}-q)^2}\nonumber\\
&= n\cdot\frac{((e^\epsilon-1)q+1)^2}{(e^\epsilon-1)^2(1-q)q}. \label{eq:uevar}
\end{align}

\mypara{Symmetric UE ($\SUE$).}
RAPPOR's implementation chooses $p$ and $q$ such that $p+q=1$; making the treatment of $1$ and $0$ symmetric.  Combining this with \eqref{eq:uereq}, we have
$$p=\frac{e^{\epsilon / 2}}{e^{\epsilon / 2} + 1},\;\;  q=\frac{1}{e^{\epsilon / 2} + 1}  $$
Plugging these into \eqref{eq:uevar}, we have
\begin{align*}
\mathrm{Var^*}[\tilde{c}_{\SUE}(i)] & = n\cdot\frac{e^{\epsilon/2}}{(e^{\epsilon/2}-1)^2}
\end{align*}

\mypara{Optimized UE ($\OUE$).}  Instead of making $p$ and $q$ symmetric, we can choose them to minimize \eqref{eq:uevar}.
Take the partial derivative of \eqref{eq:uevar} with respect to $q$, and solving $q$ to make the result $0$, we get:
{\small
\begin{align*}
\frac{\partial\left[\frac{((e^\epsilon-1)q+1)^2}{(e^\epsilon-1)^2(1-q)q}\right]}{\partial q}
=&\frac{\partial\left[\frac{1}{(e^\epsilon-1)^2}\cdot\left(\frac{(e^\epsilon-1)^2q}{1-q}+\frac{2(e^\epsilon-1)}{1-q}+\frac{1}{q(1-q)}\right)\right]}{\partial q}\\
=&\frac{\partial\left[\frac{1}{(e^\epsilon-1)^2}\cdot\left(-(e^\epsilon-1)^2+\frac{e^{2\epsilon}}{1-q}+\frac{1}{q}\right)\right]}{\partial q}\\
=&\frac{1}{(e^\epsilon-1)^2} \left(\frac{e^{2\epsilon}}{(1-q)^2}-\frac{1}{q^2}\right)=0\\
\implies&\frac{1-q}{q}=e^\epsilon, \mbox{i.e.}, q=\frac{1}{e^\epsilon+1} \mbox{ and } p=\frac{1}{2}
\end{align*}
}
Plugging $p=\frac{1}{2}$ and $q=\frac{1}{e^\epsilon+1}$ into \eqref{eq:uevar}, we get
\begin{align}
\mathrm{Var^*}[\tilde{c}_{\OUE}(i)] & = n \frac{4e^\epsilon}{(e^\epsilon-1)^2}  \label{eq:var:oue}
\end{align}

The reason why setting $p=\frac{1}{2}$ and $q=\frac{1}{e^\epsilon+1}$ is optimal when the true frequencies are small may be unclear at first glance; however, there is an intuition behind it.  When the true frequencies are small, $d$ is large.
Recall that $e^\epsilon= \frac{p}{1-p}\frac{1-q}{q}$.  Setting $p$ and $q$ can be viewed as splitting $\epsilon$ into $\epsilon_1+\epsilon_2$ such that $\frac{p}{1-p}=e^{\epsilon_1}$ and $\frac{1-q}{q}=e^{\epsilon_2}$.  That is, $\epsilon_1$ is the privacy budget for transmitting the $1$ bit, and $\epsilon_2$ is the privacy budget for transmitting each $0$ bit.  Since there are many $0$ bits and a single $1$ bit, it is better to allocate as much privacy budget for transmitting the $0$ bits as possible.  In the extreme, setting $\epsilon_1=0$ and $\epsilon_2=\epsilon$ means that setting $p=\frac{1}{2}$.

%% file: projection.tex
Both \HE and \UE use unary encoding and have $\Theta(d)$ communication cost, which is too large for some applications.
To reduce the communication cost, a natural idea is to first hash the input value into a domain of size $k<d$, and then apply the \UE method to the hashed value.  This is the basic idea underlying the \rappor{} method.  However, a problem with this approach is that two values may be hashed to the same output, making them indistinguishable from each other during decoding.  RAPPOR tries to address this in several ways.  One is to use more than one hash functions; this reduces the chance of a collision.  The other is to use cohorts, so that different cohorts use different sets of hash functions.  These remedies, however, do not fully eliminate the potential effect of collisions.  Using more than one hash functions also means that every individual bit needs to be perturbed more to satisfy $\epsilon$-LDP for the same $\epsilon$.

A better approach is to make each user belongs to a cohort by herself.  We call this the \textbf{local hashing} approach.
The random matrix-base protocol in~\cite{Bassily2015local} (described in Section~\ref{sec:over:matrix}), in its very essence, uses a local hashing encoding that maps an input value to a single bit, which is then transmitted using randomized response.
Below is the Binary Local Hashing (BLH) protocol, which is logically equivalent to the one in Section~\ref{sec:over:matrix}, but is simpler and, we hope, better illustrates the essence of the idea.  

Let $\mathbb{H}$ be a universal hash function family, such that each hash function $H\in \mathbb{H}$ hashes an input in $[d]$ into one bit.  The universal property requires that $$\forall x,y\in[d], x\ne y: \underset{H \in \mathbb{H}}{\mathsf{Pr}} [H(x)=H(y)] \leq \frac{1}{2}.$$


	


\mypara{Encoding.}  $\mathsf{Encode}_{\mathrm{BLH}}(v)= \langle H,b\rangle$, where $H\leftarrow_{R} \mathbb{H}$ is chosen uniformly at random from $\mathbb{H}$, and
$b=H(v)$.  Note that the hash function $H$ can be encoded using an index for the family $\mathbb{H}$ and takes only $O(\log d)$ bits.

\mypara{Perturbing.} $\mathsf{Perturb}_{\mathrm{BLH}}(\langle H,b\rangle)= \langle H,b'\rangle$ such that
\[
\Pr{b' = 1}  = \left\{
\begin{array}{lr}
p=\frac{e^\epsilon}{e^\epsilon + 1}, & \mbox{if} \; b = 1  \\
q= \frac{1}{e^\epsilon + 1}, & \mbox{if} \; b = 0 \\
\end{array}
\right.
\]

\mypara{Aggregation.}
$\mathsf{Support}_{\BLH}(\langle H, b\rangle)=\{v\mid H(v)=b\}$, that is, each reported $\langle H,b \rangle$ supports all values that are hashed by $H$ to $b$, which are half of the input values.  Using this $\mathsf{Support}$ function makes the protocol pure, with $p^*=p$ and $q^*=\frac{1}{2}p + \frac{1}{2}q = \frac{1}{2}$.
Plugging the values of $p^*$ and $q^*$ into \eqref{eq:absvar}, we have $$
\mathrm{Var^*}[\tilde{c}_{\mathsf{BLH}}(i)] = n\cdot\frac{(e^\epsilon+1)^2}{(e^\epsilon-1)^2}. \label{eq:blh:var}
$$

%% file: partition.tex
Once the random matrix projection protocol is cast as binary local hashing, we can clearly see that the encoding step loses information because the output is just one bit.  Even if that bit is transmitted correctly, we can get only one bit of information about the input, i.e., to which half of the input domain does the value belong.  
When $\epsilon$ is large, the amount of information loss in the encoding step dominates that of the random response step.  Based on this insight, we generalize Binary Local Hashing so that each input value is hashed into a value in $[g]$, where $g\geq 2$.  A larger $g$ value means that more information is being preserved in the encoding step.  This is done, however, at a cost of more information loss in the random response step.  As in our analysis of the Direct Encoding method, a large domain results in more information loss.   

Let $\mathbb{H}$ be a universal hash function family such that each $H\in \mathbb{H}$ outputs a value in $[g]$.

\mypara{Encoding.}  $\mathsf{Encode}(v)=\langle H, x\rangle$, where $H \in \mathbb{H}$ is chosen uniformly at random, and $x=H(v)$.

\mypara{Perturbing.}
$\mathsf{Perturb}(\langle H,x\rangle)=(\langle H,y\rangle)$, where
\[
\forall_{i \in [g]}\;\Pr{ y = i}  = \left\{
\begin{array}{lr}
p=\frac{e^\epsilon}{e^\epsilon + g - 1}, & \mbox{if} \; x = i  \\
q= \frac{1}{e^\epsilon + g - 1}, & \mbox{if} \; x \neq i \\
\end{array}
\right.
\]

\begin{theorem}[Privacy of \LH]
	The Local Hashing (\LH) Protocol satisfies $\epsilon$-LDP
\end{theorem}
\begin{proof}
	For any two possible input values $v_1,v_2$ and any output $\langle H,y\rangle $, we have,
	\begin{align*}
	\frac{\Pr{\langle H,y\rangle |v_1}}{\Pr{\langle H,y\rangle |v_2}} = \frac{\Pr{\mathsf{Perturb}(H(v_1))=y}}{\Pr{\mathsf{Perturb}(H(v_2))=y}}
	\leq \frac{p}{q}
	=e^\epsilon
	\end{align*}
\end{proof}

\mypara{Aggregation.} Let $\mathsf{Support}_{\LH}(\langle H, y\rangle )=\{i\mid H(i)=y\}$, i.e., the set of values that are hashed into the reported value.  This gives rise to a pure protocol with
\begin{align*}
p^*=p \mbox{ and }q^*=\frac{1}{g}p + \frac{g-1}{g}q = \frac{1}{g}.
\end{align*}
Plugging these values into \eqref{eq:absvar}, we have the
\begin{equation}
\mathrm{Var^*}[\tilde{c}_{\mathsf{LP}}(i)] = n\cdot\frac{(e^\epsilon-1+g)^2}{(e^\epsilon-1)^2(g-1)}. \label{eq:lp:var}
\end{equation}

\mypara{Optimized \LH {} (\OLH)}
\newcommand{\xa}{e^\epsilon}
\newcommand{\xaa}{e^{2\epsilon}}
\newcommand{\xb}{e^\epsilon-1}
\newcommand{\xbb}{(e^\epsilon-1)^2}
Now we find the optimal $g$ value, by taking the partial derivative of \eqref{eq:lp:var} with respect to $g$.
\begin{align*}
\frac{\partial \left[\frac{(e^\epsilon-1+g)^2}{(e^\epsilon-1)^2(g-1)}\right]}{\partial g}
&=\frac{\partial \left[\frac{g-1}{(e^\epsilon-1)^2}+\frac{1}{g-1}\cdot\frac{e^{2\epsilon}}{(e^\epsilon-1)^2}+\frac{2e^\epsilon}{(e^\epsilon-1)^2}\right]}{\partial g}\\
&=\frac{1}{(e^\epsilon-1)^2} -\frac{1}{(g-1)^2}\cdot\frac{e^{2\epsilon}}{(e^\epsilon-1)^2}=0\\
&\implies g=e^\epsilon+1
\end{align*}
When $g=e^\epsilon+1$, we have $p^*=\frac{e^\epsilon}{e^\epsilon + g - 1}=\frac{1}{2}$, $q^*= \frac{1}{g}=\frac{1}{e^\epsilon+1}$ into \eqref{eq:uevar}, and
\begin{equation}
\mathrm{Var^*}[\tilde{c}_{\mathsf{OLH}}(i)] = n\cdot\frac{4e^\epsilon}{(e^\epsilon-1)^2}. \label{eq:olhvar}
\end{equation}

\mypara{Comparing \OLH with \OUE.}  It is interesting to observe that the variance we derived for optimized local hashing (\OLH), i.e., \eqref{eq:olhvar} is exactly that we have for optimized unary encoding (\OUE), i.e., \eqref{eq:var:oue}.  Furthermore, the probability values $p^*$ and $q^*$ are also exactly the same.  This illustrates that \OLH and \OUE are in fact deeply connected.  \OLH can be viewed as a compact way of implementing \OUE.  
Compared with \OUE, \OLH has communication cost $O(\log n)$ instead of $O(d)$.

The fact that optimizing two apparently different encoding approaches, namely, unary encoding and local hashing, results in conceptually equivalent protocol, seems to suggest that this may be optimal (at least when $d$ is large).  However, whether this is the best possible protocol remains an interesting open question.

%% file: variances.tex
\begin{table*}[t]
	\centering
	\begin{tabular}{|c|c|c|c|c|c|c|c|}
		\hline
							& \DE & \SHE & \THE ($\theta=1$) & \SUE & \OUE    & \BLH & \OLH  \\ \hline
		Communication Cost  & $O(\log d)$   & $O(d)$ & $O(d)$& $O(d)$& $O(d)$  	  & $O(\log n)$     & $O(\log n)$ \\ \hline
        $\mathrm{Var}[\tilde{c}(i)]/n$
                            &$\frac{d-2+e^\epsilon}{(e^\epsilon -1)^2}$
                            &$\frac{8}{\epsilon^2}$
                            &$\frac{2e^{\epsilon /2}-1}{(e^{\epsilon/2}-1)^2}$
                            &$\frac{e^{\epsilon/2}}{(e^{\epsilon/2}-1)^2}$
                            &$\frac{4e^\epsilon}{(e^\epsilon-1)^2}$
                            &$\frac{(e^\epsilon+1)^2}{(e^\epsilon-1)^2}$
                            &$\frac{4e^\epsilon}{(e^\epsilon-1)^2}$\\\hline

	\end{tabular}
	\vspace{0.1cm}
	\caption{
		Comparison of Communication Cost, Computation Cost Incurred by the Aggregator,
		and Variances for different methods.
	}
	\label{tbl:cost}
\end{table*}

\begin{table*}[t]
	\centering
	\begin{tabular}{|c|c|c|c|c|c|c|c|c|c|}
		\hline
		& \DE $(d=2)$ & \DE $(d=32)$ & \DE $(d=2^{10})$ & \SHE & \THE ($\theta=1$) & \SUE & \OUE    & \BLH & \OLH
\\ \hline
$\epsilon=0.5$
&$3.92$
&$75.20$
&$2432.40$
&$32.00$
&$19.44$
&$15.92$
&$15.67$
&$16.67$
&$15.67$

\\ \hline
$\epsilon=1.0$
&$0.92$
&$11.08$
&$347.07$
&$8.00$
&$5.46$
&$3.92$
&$3.68$
&$4.68$
&$3.68$

\\ \hline
$\epsilon=2.0$
&$0.18$
&$0.92$
&$25.22$
&$2.00$
&$1.50$
&$0.92$
&$0.72$
&$1.72$
&$0.72$

\\ \hline
$\epsilon=4.0$
&$0.02$
&$0.03$
&$0.37$
&$0.50$
&$0.34$
&$0.18$
&$0.08$
&$1.08$
&$0.08$

\\ \hline

	\end{tabular}
	\vspace{0.1cm}
	\caption{
		Numerical Values of Analytical Variance of Different Methods
	}
	\label{tbl:cost_comp}
\end{table*}

\begin{figure*}[!h]
	\centering
	\subfigure[Vary $\epsilon$]{
		\label{fig:e_d}
		\includegraphics[width=0.95\columnwidth]{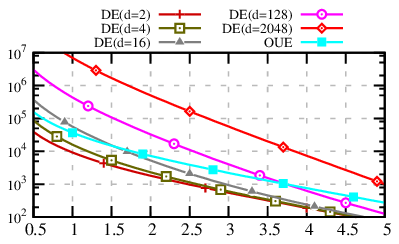}
	}
	\subfigure[Vary $\epsilon$ (fixing $d=2^{10}$)]{
		\label{fig:all_e_l_large}
		\includegraphics[width=0.95\columnwidth]{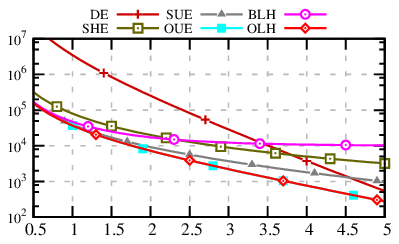}
	}	
	
	\caption{Numerical Values of Analytical Variance of Different Methods}
	\label{fig:cmp_e}
\end{figure*}

We have cast most of the LDP protocols proposed in the literature into our framework of pure LDP protocols.  Doing so also enables us to generalize and optimize existing protocols.  Now we are able to answer the question: Which LDP protocol should one use in a given setting?  

\mypara{Guideline.}
Table~\ref{tbl:cost} lists the major parameters for the different protocols.  Histogram encoding and unary encoding requires $\Theta(d)$ communication cost, and is expensive when $d$ is large.  Direct encoding and local hashing require $\Theta(\log d)$ or $\Theta(\log n)$ communication cost, which amounts to a constant in practice.  All protocols other than \DE have $O(n\cdot d)$ computation cost to estimate frequency of all values.  

Numerical values of the approximate variances using \eqref{eq:absvar} for all protocols are given in Table~\ref{tbl:cost_comp} and Figure~\ref{fig:cmp_e}.  Our analysis gives the following guidelines for choosing protocols.
\begin{itemize}
 \item
When $d$ is small, more precisely, when $d < 3e^\epsilon + 2$, \DE is the best among all approaches.

 \item
When $d > 3e^\epsilon + 2$, and the communication cost $\Theta(d)$ is acceptable, one should use \OUE.  (\OUE has the same variance as \OLH, but is easier to implement and faster because it does not need to use hash functions.)

 \item
When $d$ is so large that the communication cost $\Theta(d)$ is too large, we should use \OLH.  It offers the same accuracy as \OUE, but has communication cost $O(\log d)$ instead of $O(d)$.

\end{itemize}

\mypara{Discussions.}   In addition to the guidelines, we make the following observations.
Adding Laplacian noises to a histogram is typically used in a setting with a trusted data curator, who first computes the histogram from all users' data and then adds the noise.  \SHE applies it to each user's data.  Intuitively, this should perform poorly relative to other protocols specifically designed for the local setting.  However, \SHE performs very similarly to \BLH, which was specifically designed for the local setting.  In fact, when $\epsilon > 2.5$, \SHE performs better than \BLH.

While all protocols' variances depend on $\epsilon$, the relationships are different.   \BLH is least sensitive to change in $\epsilon$ because binary hashing loses too much information.  Indeed, while all other protocols have variance goes to $0$ when $\epsilon$ goes to infinity, \BLH has variance goes to $n$.  
\SHE is slightly more sensitive to change in $\epsilon$.  \DE is most sensitive to change in $\epsilon$; however, when $d$ is large, its variance is very high.  
\OLH and \OUE are able to better benefit from an increase in $\epsilon$, without suffering the poor performance for small $\epsilon$ values. 
Another interesting finding is that when $d=2$, the variance of \DE is $\frac{e^\epsilon}{(e^\epsilon -1)^2}$, which is exactly $\frac{1}{4}$ of that of \OUE and \OLH, whose variances do not depend on $d$.  
Intuitively, it is easier to transmit a piece of information when it is binary, i.e., $d=2$.  As $d$ increases, one needs to ``pay'' for this increase in source entropy by having higher variance.  However, it seems that there is a cap on the ``price'' one must pay no matter how large $d$ is, i.e., \OLH's variance does not depend on $d$ and is always $4$ times that of \DE with $d=2$.  There may exist a deeper reason for this rooted in information theory. Exploring these questions is beyond the scope of this paper.

\begin{figure*}[t]
	\centering
	\subfigure[Vary $d$ (fixing $\epsilon=4$)]{
		\label{fig:match1_l_e=04}
		\includegraphics[width=0.95\columnwidth]{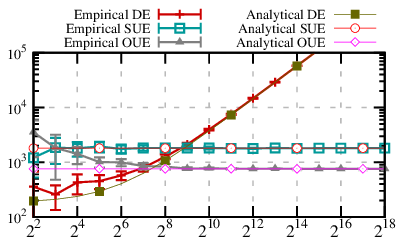}
	}
	\subfigure[Vary $d$ (fixing $\epsilon=4$)]{
		\label{fig:match2_l_e=4}
		\includegraphics[width=0.95\columnwidth]{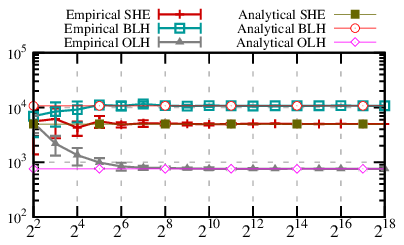}
	}	
	
%
	\subfigure[Vary $\epsilon$ (fixing $d=2^{10}$)]{
		\label{fig:match1_e_l=10}
		\includegraphics[width=0.95\columnwidth]{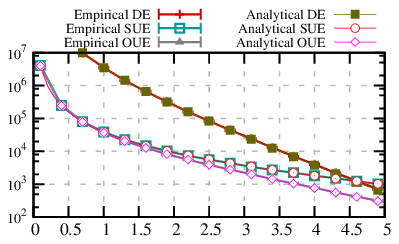}
	}
	\subfigure[Vary $\epsilon$ (fixing $d=2^{10}$)]{
		\label{fig:match2_e_l=10}
		\includegraphics[width=0.95\columnwidth]{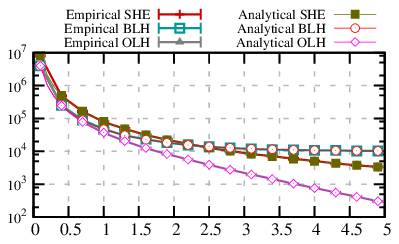}
	}	
	
	\caption{Comparing Empirical and Analytical Variance.}
	\label{fig:empvsana}
\end{figure*}

%% file: eval.tex
We empirically evaluate these protocols on both synthetic and real-world datasets. All experiments are run ten times and we plot the mean and standard deviation.

\subsection{Verifying Correctness of Analysis}

The conclusions we drew above are based on analytical variances.  We now show that our analytical results of variances match the empirically measured squared errors.
For the empirical data, we issue queries using the protocols and measure the average of the squared errors, namely, $\frac{1}{d}\sum_{i\in[d]}\left[\tilde{c}(i)-nf_i\right]^2$, where $f_i$ is the fraction of users taking value $i$.  We run queries for all $i$ values and repeat for ten times.  We then plot the average and standard deviation of the squared error.  We use synthetic data generated by following the Zipf's distribution (with distribution parameter $s=1.1$ and $n=10,000$ users), similar to experiments in~\cite{rappor}.

Figure~\ref{fig:empvsana} gives the empirical and analytical results for all methods.  In Figures~\ref{fig:match1_l_e=04} and~\ref{fig:match2_l_e=4}, we fix $\epsilon=4$ and vary the domain size.  For sufficiently large $d$ (e.g., $d\geq 2^6$), the empirical results match very well with the analytical results.  When $d<2^6$, the analytical variance tends to underestimate the variance, because in \eqref{eq:absvar} we ignore the $f_i$ terms.  Standard deviation of the measured squared error from different runs also decreases when the domain size increases. In Figures~\ref{fig:match1_e_l=10} and~\ref{fig:match2_e_l=10}, we fix the domain size to $d=2^{10}$ and vary the privacy budget.  We can see that the analytical results match the empirical results for all $\epsilon$ values and all methods.

\subsection{Towards Real-world Estimation}

We run \OLH, \BLH, together with \rappor, on real datasets. The goal is to understand how does each protocol perform in real world scenarios and how to interpret the result. Note that \rappor does not fall into the pure framework of LDP protocols so we cannot use Theorem \ref{thm:pure_var} to obtain the variance analytically. Instead, we run experiments to examine its performance empirically. 
We implement \rappor in Python. The regression part, which \rappor introduces to handle the collisions generated by the hash functions, is implemented using Scikit-learn library~\cite{sklearn}. We use $128$-bit Bloom filter, 2 hash functions and 8 and 16 cohorts in \rappor.

\mypara{Datasets.}
We use the \textbf{Kosarak} dataset~\cite{Kosarak}, which contains the click stream of a Hungarian news website.
There are around $8$ million click events for $41,270$ different pages.
The goal is to estimate the popularity of each page, assuming all events are reported.

\subsubsection{Accuracy on Frequent Values}
One goal of estimating a distribution is to find out the frequent values and accurately estimate them.
We run different methods to estimate the distribution of the Kosarak dataset. After the estimation, we issue queries for the 30 most frequent values in the original dataset. We then calculate the average squared error of the 30 estimations produced by different methods. Figure~\ref{fig:varkosarak} shows the result. 
We try \rappor with both 8 cohorts (RAP(8)) and 16 cohorts (RAP(16)). It can be seen that 
when $\epsilon>1$, \OLH starts to show its advantage. Moreover, variance of \OLH decreases fastest among the four. Due to the internal collision caused by Bloom filters, the accuracy of \rappor does not benefit from larger $\epsilon$. We also perform this experiment on different datasets, and the results are similar.

\input{real_img.tex}
\subsubsection{Distinguish True Counts from Noise}

Although there are noises, infrequent values are still unlikely to be estimated to be frequent.
Statistically, the frequent estimates are more reliable, because the probability it is generated from an infrequent value is quite low.
However, for the infrequent estimates, we don't know whether it comes from an originally infrequent value or a zero-count value.
Therefore, after getting the estimation, we need to choose which estimate to use, and which to discard.

\mypara{Significance Threshold.}
In~\cite{rappor}, the authors propose to use the significance threshold. After the estimation, all estimations above the threshold are kept, and those below the threshold $T_s$ are discarded.
$$T_s = \Phi^{-1}\left(1-\frac{\alpha}{d}\right) \sqrt{\mathrm{Var^*}},$$
where $d$ is the domain size, $\Phi^{-1}$ is the inverse of the cumulative density function of standard normal distribution, and the term inside the square root is the variance of the protocol.
Roughly speaking, the parameter $\alpha$ controls the number of values that originally have low frequencies but estimated to have frequencies above the threshold (also known as false positives). We use $\alpha=0.05$ in our experiment.

For the values whose estimations are discarded, we don't know for sure whether they have low or zero frequencies. Thus, a common approach is to assign the remaining probability to each of them uniformly.

Recall $\mathrm{Var^*}$ is the term we are trying to minimize. So a protocol with small estimation variance for zero-probability values will have a lower threshold, and thus be able to detect more values reliably.

\mypara{Number of Reliable Estimation.}
We run different protocols using the significance threshold $T_s$ on the Kosarak dataset. Note that $T_s$ will change as $\epsilon$ changes. We define a true (false) positive as a value that has frequency above (below) the threshold, and is estimated to have frequency above the threshold. In Figure~\ref{fig:kosarak_vary_e}, we show the number of true positives versus $\epsilon$. As $\epsilon$ increases, the number of true positives increases. When $\epsilon=4$, \rappor can output 75 true positives, \BLH can only output 36 true positives, but \OLH can output nearly 200 true positives.
We also note that the output sizes are similar for \rappor and \OLH, which indicates that \OLH gives out very few false positives compared to \rappor. The cohort size does not affect much in this setting.

\subsubsection{On Information Quality}
Now we test both the number of true positives and false positives, varying the threshold. We run \OLH, \BLH and \rappor on the Kosarak dataset. 

As we can see in Figure~\ref{fig:rockyoutp}, fixing a threshold, \OLH and \BLH performs similarly in identifying true positives, which is as expected, because frequent values are rare, and variance does not change much the probability it is identified. \rappor performs slightly worse because of the Bloom filter collision.

As for the false positives, as shown in Figure~\ref{fig:rockyoufp}, different protocols perform quite differently in eliminating false positives. When fixing $T_s$ to be $5,000$, \OLH produces tens of false positives, but \BLH will produce thousands of false positives. The reason behind this is that, for the majority of infrequent values, their estimations are directly related to the variance of the protocol. A protocol with a high variance means that more infrequent values will become frequent during estimation.
As a result, because of its smallest $\mathrm{Var^*}$, \OLH produces the least false positives while generating the most true positives.

%% file: real_img.tex
\begin{figure}[t]
    \centering
    
    \includegraphics[width=0.95\columnwidth]{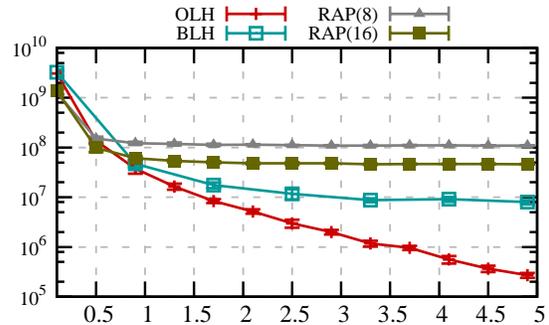}
    
    \caption{
        Average squared error, varying $\epsilon$.}
    \label{fig:varkosarak}
\end{figure}

\begin{figure}[!h]
	\centering
	
	\includegraphics[width=0.95\columnwidth]{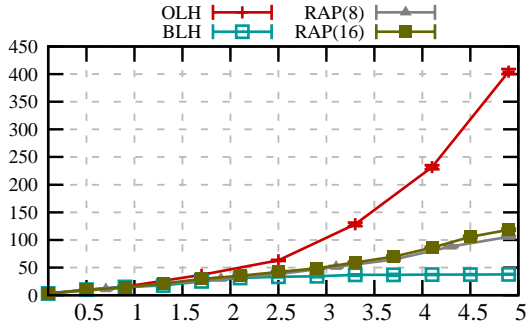}
	
	\caption{
		Number of true positives, varying $\epsilon$, using significance threshold. The dashed line corresponds to the average number of items identified.}
	\label{fig:kosarak_vary_e}
\end{figure}

\begin{figure*}[!h]
    \centering
    \centering

    \centering
    \subfigure[Number of True Positives]{
        \label{fig:rockyoutp}
        \includegraphics[width=0.95\columnwidth]{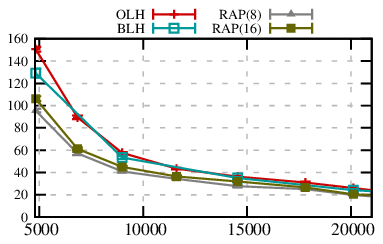}
    }
    \subfigure[Number of False Positives]{
        \label{fig:rockyoufp}
        \includegraphics[width=0.95\columnwidth]{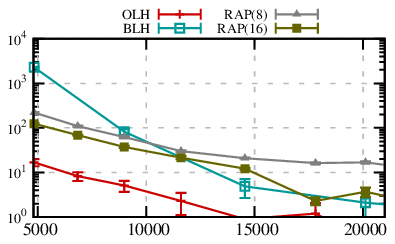}
    }
    \caption{
        Results on Kosarak dataset.
        The $y$ axes are the number of identified hash values that is true/false positive. The $x$ axes are the threshold. We assume $\epsilon=4$.}
    \label{fig:rockyou}
\end{figure*}

%
%
%
%
%

%% file: related.tex
The notion of differential privacy and the technique of adding noises sampled from the Laplace distribution were introduced in~\cite{DMNS06}.  Many algorithms for the centralized setting have been proposed.  See~\cite{dpbook} for a theoretical treatment of these techniques, and~\cite{Li2016book} for a treatment from a more practical perspective.  It appears that only algorithms for the LDP settings have seen real world deployment.
Google deployed \rappor~\cite{rappor} in Chrome, and Apple~\cite{Apple} also uses similar methods to help with predictions of spelling and other things.

State of the art protocols for frequency estimation under LDP are RAPPOR by Erlingsson et al.~\cite{rappor} and Random Matrix Projection (\BLH) by Bassily and Smith~\cite{Bassily2015local}, which we have presented in Section~\ref{sec:background} and compared with in detail in the paper.  These protocols use ideas from earlier work~\cite{mishra2006privacy,Duchi:2013:LPS}.  Our proposed Optimized Unary Encoding (\OUE) protocol builds upon the Basic RAPPOR protocol in~\cite{rappor}; and our proposed Optimized Local Hashing (\OLH) protocol is inspired by \BLH in~\cite{Bassily2015local}.
Wang et al.~\cite{wang2016private} uses both generalized random response (Section~\ref{sec:direct}) and Basic RAPPOR for learning weighted histogram.
Some researchers use existing frequency estimation protocols as primitives to solve other problems in LDP setting.  
For example, Chen et al.~\cite{samsung_location} uses \BLH~\cite{Bassily2015local} to learn location information about users.
Qin et al.~\cite{ccs16} use RAPPOR~\cite{rappor} and \BLH~\cite{Bassily2015local} to estimate frequent items where each user has a set of items to report.
These can benefit from the introduction of \OUE and \OLH in this paper.  

There are other interesting problems in the LDP setting beyond frequency estimation.  In this paper we do not study them.  One problem is to identify frequent values when the domain of possible input values is very large or even unbounded, so that one cannot simply obtain estimations for all values to identify which ones are frequent.  This problem is studied in~\cite{hsu2012distributed,Bassily2015local,rappor2}.  Another problem is estimating frequencies of itemsets~\cite{EGS03,ESA+02}.  {Nguy{\^e}n et al.~\cite{samsung} studied how to report numerical answers (e.g., time of usage, battery volume) under LDP.  
When these protocols use frequency estimation as a building block (such as in~\cite{rappor2}), they can directly benefit from results in this paper.  Applying insights gained in our paper to better solve these problems is interesting future work.


Kairouz et al.~\cite{kairouz2014extremal} study the problem of finding the optimal LDP protocol for two goals: (1) hypothesis testing, i.e., telling whether the users' inputs are drawn from distribution $P_0$ or $P_1$, and (2) maximize mutual information between input and output.  We note that these goals are different from ours.  Hypothesis testing does not reflect dependency on $d$.   Mutual information considers only a single user's encoding, and not aggregation accuracy.  For example, both global and local hashing have exactly the same mutual information characteristics, but they have very different accuracy for frequency estimation, because of collisions in global hashing.  Nevertheless, it is found that for very large $\epsilon$'s, \protRandom{} is optimal, and for very small $\epsilon$'s, \BLH is optimal.  This is consistent with our findings.  However, analysis in~\cite{kairouz2014extremal} did not lead to generalization and optimization of binary local hashing, nor does it provide concrete suggestion on which method to use for a given $\epsilon$ and $d$ value.

%% file: conc.tex
In this paper, we study frequency estimation in the Local Differential Privacy (LDP) setting.  We have introduced a framework of pure LDP protocols together with a simple and generic aggregation and decoding technique.  This framework enables us to analyze, compare, generalize, and optimize different protocols, significantly improving our understanding of LDP protocols.
More concretely, we have introduced the Optimized Local Hashing (\OLH) protocol, which has much better accuracy than previous frequency estimation protocols satisfying LDP.
We provide a guideline as to which protocol to choose in different scenarios. Finally we demonstrate the advantage of the \OLH in both synthetic and real-world datasets.
For future work, we plan to apply the insights we gained here to other problems in the LDP setting, such as discovering heavy hitters when the domain size is very large. 

%% file: addi_eval.tex
This section provides additional experimental evaluation results. We first try to measure average squared variance on other datasets. Although \rappor did not specify a particular optimal setting, we vary the number of cohorts and find differences. In the end, we use different privacy budget on the Rockyou dataset.

\begin{figure*}[!h]
	\centering
	\subfigure[Zipf's Top 100 Values]{
		\label{fig:zipfs100}
		\includegraphics[width=0.95\columnwidth]{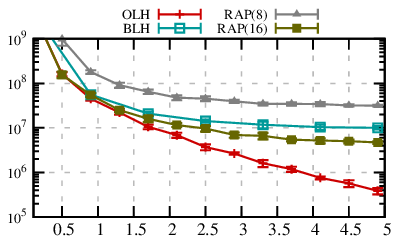}
	}
	\subfigure[Normal Top 100 Values]{
		\label{fig:normal100}
		\includegraphics[width=0.95\columnwidth]{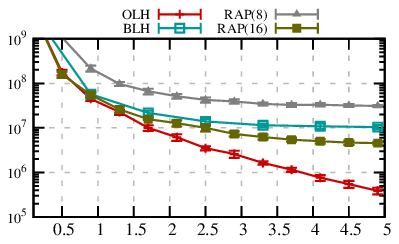}
	}
	
	\centering
	\subfigure[Zipf's Top 50 Values]{
		\label{fig:zipfs50}
		\includegraphics[width=0.95\columnwidth]{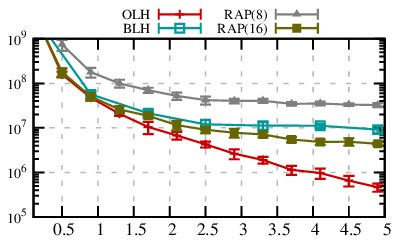}
	}
	\subfigure[Zipf's Top 10 Values]{
		\label{fig:zipfs10}
		\includegraphics[width=0.95\columnwidth]{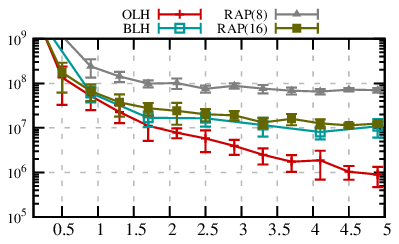}
	}
	
	\caption{
		Average squared errors on estimating a distribution of $10000000$ points. \rappor is used with $128$-bit long Bloom filter and 2 hash functions.}
	\label{fig:syn}
	
\end{figure*}
\begin{figure}[t]
	\centering
	
	\includegraphics[width=0.95\columnwidth]{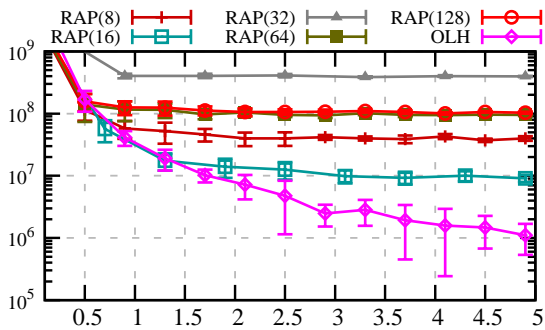}
	
	\caption{
		Average squared error on estimating a normal distribution of $1000000$ points. \rappor is used with $128$-bit long Bloom filter and 2 hash functions.}
	\label{fig:rappor}
\end{figure}

\subsection{Effect of Cohort Size}

In~\cite{rappor}, the authors did not identify the best cohort size to use. Intuitively, if there are too few cohorts, many values will be hashed to be the same in the Bloom filter, making it difficult to distinguish these values. If there are more cohorts, each cohort cannot convey enough useful information. 
Here we try to test what cohort size we should use. 
We generate $10000000$ values following the Zipf's distribution (with parameter 1.5), but only use the first 128 most frequent values because of memory limitation caused by regression part of \rappor. We then run \rappor using 8, 16, 32, and 64, and 128 cohorts. We measure the average squared errors of queries about the top 10 values, and the results are shown in Figure~\ref{fig:rappor}. As we can see, more cohorts does not necessarily help lower the squared error because the reduced probability of collision within each cohort. But it also has the disadvantage that each cohort may have insufficient information. 
It can be seen \OLH still performs best.

\begin{figure*}[!t]
	\centering
	\subfigure[Number of True Positives $\epsilon=2$]{
		\label{fig:rockyoutpe2}
		\includegraphics[width=0.95\columnwidth]{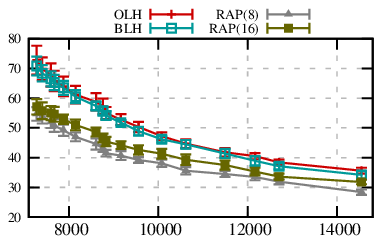}
	}
	\subfigure[Number of False Positives $\epsilon=2$]{
		\label{fig:rockyoufpe2}
		\includegraphics[width=0.95\columnwidth]{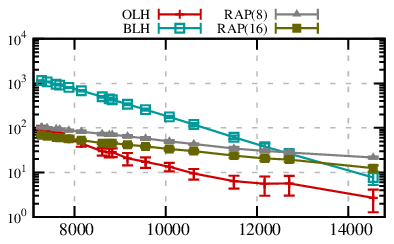}
	}
	
	\centering
	\subfigure[Number of True Positives $\epsilon=1$]{
		\label{fig:rockyoutpe1}
		\includegraphics[width=0.95\columnwidth]{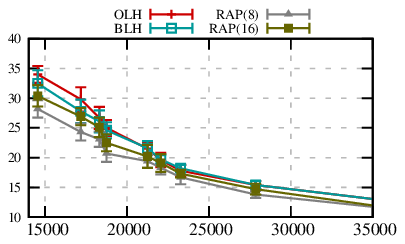}
	}
	\subfigure[Number of False Positives $\epsilon=1$]{
		\label{fig:rockyoufpe1}
		\includegraphics[width=0.95\columnwidth]{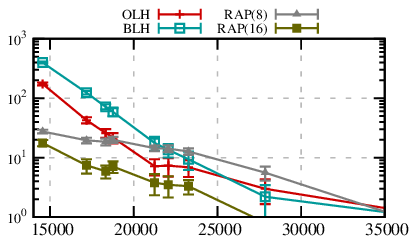}
	}
	\caption{
		Results on Kosarak dataset for $\epsilon=2$ and $1$.
		The $y$ axes are the number of identified hash values that is true/false positive. The $x$ axes are the threshold.}
	\label{fig:rockyou_other}
\end{figure*}

\subsection{Performance on Synthetic Datasets}
In Figure~\ref{fig:syn}, we test performance of different methods on synthetic datasets. We generate $10000000$ points (rounded to integers) following a normal distribution (with mean 500 and standard deviation 10) and a Zipf's distribution (with parameter 1.5). The values range from 0 to 1000. We then test the average squared errors on the most frequent 100 values. It can be seen that different methods performs similarly in different distributions. \rappor using 16 cohorts performs better than \BLH. This is because, when the number of cohort is enough, each user kind of has his own hash functions. This can be viewed as a kind of local hashing function. When we only test the top 10 values instead of top 50, RAP(16) and \BLH perform similarly. Note that \OLH still performs best among all distributions.

\subsection{Effect of $\epsilon$ on Information Quality}
We also run experiments on the Kosarak dataset setting $\epsilon=2$ and $1$. It can be seen that as $\epsilon$ decreases, the number of false positives increases for the same threshold. This is as expected, as variance will increase for every estimated value. On the other hand, the number of true positives does not change much. For instance, for the \OLH method, at threshold $T=10000$, we can recover nearly $50$ true positive values both when $\epsilon$ is $2$ and $4$. However, there are $14$ instead of $3$ false positives when $\epsilon=2$ compared to when $\epsilon=4$.
The advantage of \OLH in the case when $\epsilon=1$ is not significant.
RAP(16) consistently performs better than RAP(8) in this case.

%
%
%

